\documentclass[11pt]{article}
\usepackage{fullpage}
\usepackage[psamsfonts]{eucal}
\usepackage{amsmath}
\usepackage{amsfonts}
\usepackage{amssymb}
\usepackage{amstext}
\usepackage{amscd}
\usepackage{amsthm}
\usepackage{makeidx}
\usepackage{graphicx}
\usepackage[hidelinks]{hyperref}
\usepackage{supertabular}
\usepackage{enumerate}
\usepackage{titling}
\usepackage{wasysym}
\usepackage{comment}
\usepackage{cleveref}
\usepackage{listings}

\usepackage{microtype}
\usepackage{abbreviations}

\theoremstyle{plain}
\newtheorem{thm}{Theorem}

\newtheorem{prop}[thm]{Proposition}

\newtheorem{dfn}{Definition}
\newtheorem{fact}[thm]{Fact}

\numberwithin{thm}{section}

\begin{document}
\title{The necessity of conditions for graph quantum ergodicity and Cartesian products with an infinite graph}
\author{Theo McKenzie\thanks{\texttt{mckenzie@math.berkeley.edu.} Supported by NSF Grant DGE-1752814.}\\ University of California, Berkeley}
\date{\today}
\maketitle
\begin{abstract}
    Anantharaman and Le Masson proved that any family of eigenbases of the adjacency operators of a family of graphs is quantum ergodic (a form of delocalization) assuming the graphs satisfy conditions of expansion and high girth. In this paper, we show that neither of these two conditions is sufficient by itself to necessitate quantum ergodicity. We also show that having conditions of expansion and a specific relaxation of the high girth constraint present in later papers on quantum ergodicity is not sufficient. We do so by proving new properties of the Cartesian product of two graphs where one is infinite.
\end{abstract}

\section{Introduction}
\newcommand{\G}{\mathcal{G}}
\newcommand{\A}{\mathcal A}
\newcommand{\B}{\mathcal B}

While classical integral systems often have periodic orbits in phase space, eigenstates of quantized chaotic systems tend to be uniformly distributed (see \cite{zelditch2005quantum}). This phenomenon is expressed in the high energy limit through the concept of {quantum ergodicity}. Consider a compact Riemannian manifold $(M,g)$ and a basis of eigenfunctions $\{\psi_j\}$ of the Laplace-Beltrami operator $\Delta$ on $M$ with eigenvalues $\{\lambda_j\}$. We say $\{\psi_j\}$ is \emph{quantum ergodic} if for every continuous test function $a:M\rightarrow \R$,

\begin{equation*}
\lim_{\lambda\rightarrow +\infty} \frac{1}{N(\lambda)}\sum_{\lambda_j\leq \lambda}\left|\langle \psi_j,a \psi_j\rangle-\int_M a~d\textnormal{Vol}\right|^2=0.
\end{equation*}
Here $\langle \psi_j,a\psi_j\rangle:=\int_M a(x)|\psi(x)|^2 ~d \textnormal{Vol}(x)$ and $N(\lambda):=|\{\lambda_j\leq \lambda\}|$. Shnirelman's Theorem \cite{shnirel1974ergodic, de1985ergodicite, zelditch1987uniform} states that if the geodesic flow of $M$ is ergodic with respect to the Liouville measure, then $\{\psi_j\}$ is quantum ergodic.

Discrete graphs have provided a fruitful model for quantum chaos \cite{kottos1997quantum, kottos1999periodic}, and Brooks and Lindenstrauss initiated the study of conditions of localization and delocalization of eigenvectors on large regular discrete graphs \cite{brooks2013non}. They proved that if small sets in a graph expand well (for example the graph has high girth), all eigenvectors are delocalized in a quantifiable way depending on this expansion.

It is in this context that Anantharaman and Le Masson proved a result on discrete graphs analogous to Shnirelman's Theorem \cite{anantharaman2015quantum}. To introduce this result, we consider an infinite family of $d$-regular graphs $(G_n)=(V_n,E_n)$ with $d$ constant and $|V_n|=n$. We write $\A_n$ to denote the adjacency operator of $G_n$. $d=\lambda_1\geq \lambda_2\geq \cdots \geq \lambda_n$ are the eigenvalues of $\A_n$. We also require the following definitions.

\begin{dfn}  The family of graphs $(G_n)$ is said to satisfy \emph{EXP} if there is a constant $\epsilon>0$ such that for each $\A_n$, $\max\{\lambda_2,|\lambda_n|\}\leq (1-\epsilon)d$.
\end{dfn}

\begin{dfn}[\cite{benjamini2011recurrence}]
 Take $\mu$ to be a measure over isomorphism classes of rooted, potentially infinite graphs. The family of graphs $(G_n)$ is said to have Benjamini-Schramm limit $\mu$ if for each fixed $R>0$, as $n\rightarrow\infty$, the distribution of isomorphism classes of rooted balls of radius $R$ in $G_n$ around a root selected from $V_n$ uniformly at random converges weakly to the distribution of isomorphism classes of balls of radius $R$ around the roots of graphs according to $\mu$. For a graph $H$, we say $(G_n)$ has unrooted Benjamini-Schramm limit $H$ if the limiting measure $\mu$ is $H$ with a root of $H$ selected uniformly at random.
\end{dfn}

\begin{dfn}
The family of graphs $(G_n)$ is said to satisfy \emph{BST} if it has unrooted Benjamini-Schramm limit $T_d$, where $T_d$ is the infinite $d$-regular tree. Namely, for all fixed $R>0$, \[\lim_{n\rightarrow\infty}\frac{|\{x\in V_n, \rho(x)<R\}|}{n}\rightarrow 0,\] where $\rho(x)$ is the injectivity radius of $x$ (the largest $r$ such that the ball of radius $r$ around $x$ is a tree). 
\end{dfn}

These properties together are enough to guarantee quantum ergodicity.
\begin{thm}[\cite{anantharaman2015quantum} Theorem 1] \label{thm:qe}
Assume that $(G_n)=(V_n,E_n)$ is a family of graphs that satisfies EXP and BST. Let $a_n: V_n\rightarrow \R$ be series of functions such that $\sum_{v\in V_n} a_n(v)=0$ and $\|a_n\|_\infty \leq 1$. Then for any series of orthonormal eigenbases $(\psi_1^{(n)},\ldots, \psi_n^{(n)})$ of $(\A_n)$,
\begin{equation}\label{eq:QE}
\lim_{n\rightarrow\infty}\frac{1}{n}\sum_{i=1}^{n} |\langle \psi_i^{(n)},a_n\psi_i^{(n)}\rangle|^2=0,
\end{equation}
where
\[
\langle \psi_i^{(n)},a_n\psi_i^{(n)}\rangle=\sum_{v\in V_n} a_n(v)|\psi_i^{(n)}(v)|^2.
\]
\end{thm}
A series of eigenvectors that satisfies (\ref{eq:QE}) is called \emph{quantum ergodic}. In fact, the theorem can be generalized to more general operators $a_n$ than given, but the above formulation is sufficient for our purposes. 

 Anantharaman and Le Masson suggest that the EXP condition is analogous to the requirement of ergodicity in Shnirelman's Theorem. Therefore, it is natural to wonder whether EXP alone is sufficient to necessitate quantum ergodicity, as no other assumptions are made in Shnirelman's Theorem. However, we show that this is not the case.

\begin{thm}\label{thm:mainexp}
There is an infinite family of graphs $(G'_n)$ satisfying EXP that have a family of orthonormal eigenbases of the adjacency operators $(\psi_1^{(n)},\ldots, \psi_n^{(n)})$ that violates quantum ergodicity. Specifically, there is a series of functions $a_n:V_n\rightarrow \R_n$, $\|a_n\|_\infty\leq 1$, $\sum_{v\in V_n} a_n(v)=0$ such that for each $n$,

\[
\frac{1}{n}\sum_{i=1}^{n} |\langle \psi_i^{(n)},a_n\psi_i^{(n)}\rangle|^2= 1/2.
\]
\end{thm}

The family of graphs in Theorem \ref{thm:mainexp} is $(G_n\square C_4)$ for any family of graphs $(G_n)$ that satisfies EXP. Here $G_1\square G_2$ denotes the Cartesian product of graphs $G_1$ and $G_2$, and $C_4$ is the cycle graph of length 4. Intuitively, EXP measures expansion at global scales, whereas the Cartesian product creates a pattern on a local scale that causes localization of eigenvectors. The Cartesian product is particularly useful because of the explicit formula of its eigenvectors based on the eigenvectors of the two original graphs. Therefore, because $C_4$ has an eigenbasis with localized eigenvectors, the series of graphs $(G_n\square C_4)$ all have many localized eigenvectors. In fact, $C_4$ can be replaced with any graph with an adjacency operator with localized eigenvectors. Moreover, $(G_n\square C_4)$ satisfies EXP because of the relationship between eigenvalues of the adjacency operator of the Cartesian product with those of the adjacency operators of the original graphs. For the various properties of the Cartesian product, see Section 2.3 of Cvetkovi{\'c}, Rowlinson, and Simic \cite{cvetkovic1997eigenspaces}. 

\begin{figure}
    \centering
    \includegraphics[width=16cm]{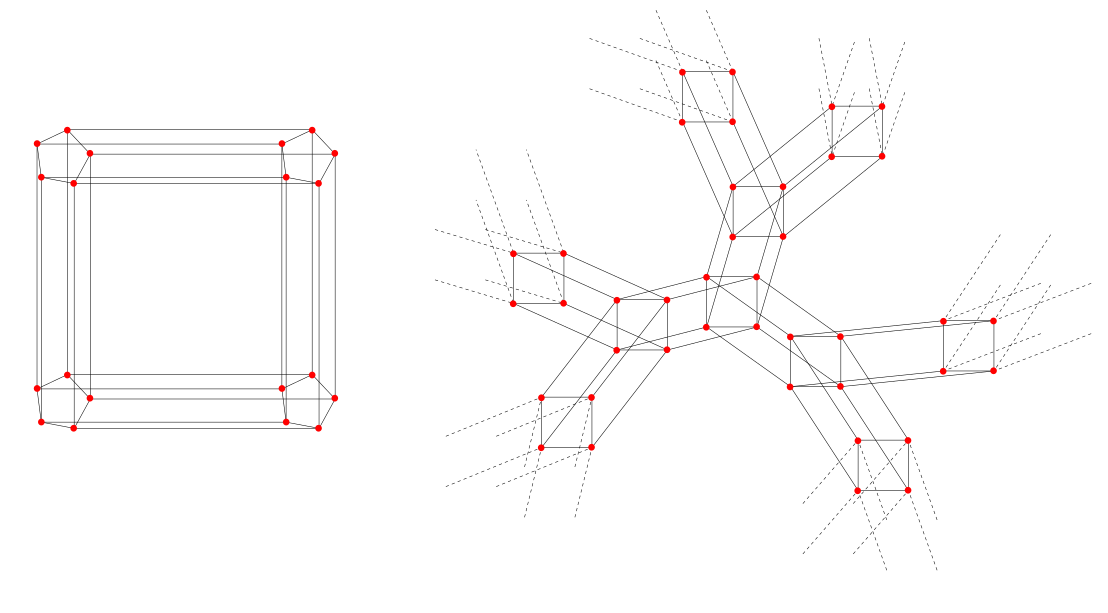}
    \caption{$C_4\square C_5$ and a portion of $T_3\square C_4$. The Cartesian product $G_1\square G_2$ can be thought of as replacing each vertex of $G_1$ with a copy of $G_2$. Note that $G_1\square G_2\cong G_2\square G_1$.}
    \label{fig:bslimit}
\end{figure}

Considering we cannot fully remove the requirement of BST, we then try to relax it. In order to necessitate quantum ergodicity in Schr\"odinger operators \cite{anantharaman2019quantum} and quantum graphs \cite{anantharaman2021quantum}), along with requirements similar to EXP and BST, an extra requirement is added that the imaginary part of the entries of the Green's function of the Benjamini Schramm limit is bounded for all $z\in \C^+$, where $\C^+$ is the upper half of the complex plane. For the adjacency operator, this property is a generalization of BST, as the Green's function of the infinite tree is known to be have bounded imaginary part for all $z\in \C^+$ (see \cite{warzel2013resonant} for a proof). Therefore we asked whether we could relax BST to a condition bounding the imaginary part of entries of the Green's function of the Benjamini-Schramm limit.

The first step is to calculate the Benjamini Schramm limit of $(G_n\square X)$. Of course, by Theorem \ref{thm:qe}, $(G_n\square C_4)$ cannot satisfy BST. In fact, the Cartesian product creates many cycles at every vertex. We show that the Benjamini-Schramm limit commutes with the Cartesian product, as for any graph $X$, the sequence of graphs $(G_n\square X)$ converges to the Cartesian product of the Benjamini-Schramm limit of $(G_n)$ with $X$. Therefore if our family of $d$-reguler graphs $(G_n)$ satisfies BST, then the Benjamini-Schramm limit of $(G_n\square C_4)$ is $T_d\square C_4$.

Examining the entries of the Green's function in our example, we show that that for an infinite graph $G_1$ and finite $G_2$, the Green's function of $G_1\square G_2$ follows the pattern of the spectrum of the Cartesian product of finite graphs. Namely, we prove the following, which could be of independent interest. Here $\G_G^z$ denotes the Green's function of $G$ at $z$.

  \begin{thm}\label{thm:green}
  Consider a (potentially infinite) graph $G_1$ and a finite graph $G_2$ with adjacency operators $\A_1$ and $\A_2$, respectively. Let $\psi_1,\ldots \psi_k$ be an orthonormal eigenbasis of $\A_2$ with eigenvalues $\lambda_1,\ldots, \lambda_k$.
  
We have

\[
\G_{G_1\square G_2}^{z}=\sum_{i=1}^{k} \G_{G_1}^{z-\lambda_i}\otimes \psi_i\psi_i^T.
\]
\end{thm}

Therefore the entries of the Green's function of $T_d\square C_4$ can be written as a linear combination of entries of Green's functions on $T_d$. As this latter quantity has bounded imaginary part everywhere, $\G_{T_d\square C_4}^z$ also has bounded imaginary part. This means that by taking $(G_n)$ to satisfy both EXP and BST, the family of graphs $(G_n\square C_4)$ satisfies EXP and has bounded imaginary part of the Green's function in the Benjamini-Schramm limit, but nevertheless by Theorem \ref{thm:mainexp} it violates quantum ergodicity. Therefore, in general, BST cannot be generalized to the requirement of having bounded imaginary part of the Green's function.

\begin{figure}
    \centering
    \includegraphics[width=9cm]{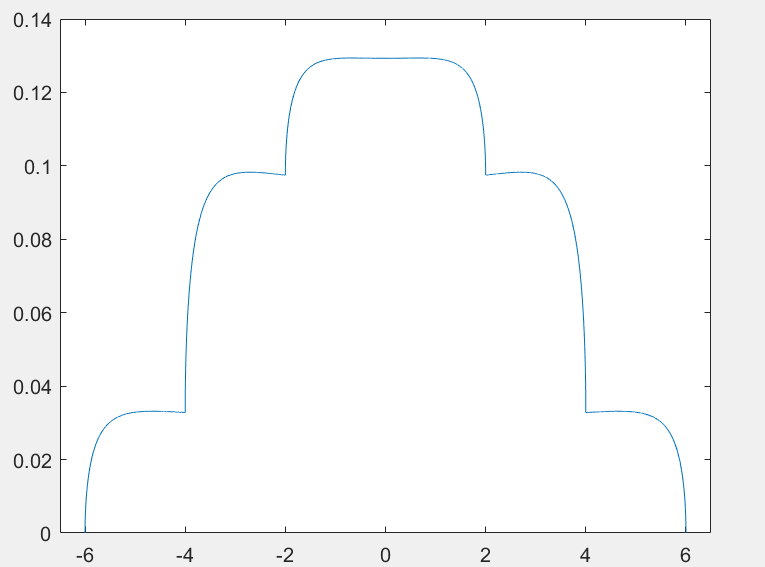}
    \caption{A plot of the spectral density of $T_5\square C_4$. It is the sum of the Kesten-McKay measure shifted by the different eigenvalues of $\A_{C_4}$.}
    \label{fig:my_label}
\end{figure}
We end the paper with Section \ref{sec:onlybst}, which shows that BST by itself is not sufficient to necessitate quantum ergodicity.
\begin{thm}\label{thm:mainbst}
There is an infinite family of graphs $(H_n)$ satisfying BST that have a family of orthonormal eigenbases of the adjacency operators $(\psi_1^{(n)},\ldots, \psi_n^{(n)})$ that violates quantum ergodicity. Specifically, there is a series of functions $a_n:V_n\rightarrow \R_n$, $\|a_n\|_\infty\leq 1$, $\sum_{v\in V_n} a_n(v)=0$ such that for each $n$,

\[
 \frac{1}{n}\sum_{i=1}^{n} |\langle \psi_i^{(n)},a_n\psi_i^{(n)}\rangle|^2\geq 1/d.
\]
\end{thm}

As EXP measures global expansion, BST measures local expansion, so our example creates localization by creating patterns on a global scale. Our construction is similar to those of \cite{ganguly2021non, alon2019high, mckenzie2020high} in that we take a set of high girth graphs and connect them in such a way that creates a geometric phenomenon without destroying girth.

\subsection{Related work}
\paragraph{Eigenvectors of random regular graphs.}
The shape of eigenvectors of random regular graphs has received considerable attention. Bauerschmidt, Huang, and Yao and then Huang and Yao \cite{bauerschmidt2019local, huang2021spectrum} proved bounds on the entries of the Green's function that hold with high probability. These are sharp enough to in turn imply a $\polylog(n)/\sqrt{n}$ infinity norm bound on any normalized eigenvector, as well as show quantum \emph{unique} ergodicity, a notion of localization stronger than quantum ergodicity. In another vein, Backhausz and Szegedy proved that the distribution of entries of an eigenvector of a random regular graph must approximate a Gaussian distribution by showing that the only typical eigenvector processes are Gaussian waves \cite{backhausz2019almost} (note that this does not necessitate delocalization, as the distribution of entries of a localized vector is close to a Gaussian of variance 0).

\paragraph{Other results in graph quantum ergodicity.}
For an overview of results before 2019, see \cite{anantharaman2019recent}. Since the original proof of quantum ergodicity, Anantharaman \cite{anantharaman2017quantum} and Brooks, Le Masson, and Lindenstrauss \cite{brooks2016quantum} have given alternate proofs of Theorem \ref{thm:qe}. Quantum ergodicity statements have since then been found for a variety of graphical models, including quantum graphs \cite{ingremeau2020quantum, anantharaman2021quantum} and the Anderson model on the Bethe lattice \cite{anantharaman2017quantum2}.

\paragraph{Other eigenvector delocalization results.}
Ganguly and Srivastava built upon the result of Brooks and Lindenstrauss to achieve a relation between local expansion and delocalization that is essentially tight \cite{ganguly2021non}. Alon, Ganguly and Srivastava then further examined the example showing tightness, increasing the lower bound on girth and showing it is a good spectral expander \cite{alon2019high}.

\paragraph{Nodal domains of random graphs.}
Another line of research has asked about the shape of the subgraph created by taking an eigenvector $f$ and deleting all edges $(u,v)$ such that $f(u)f(v)<0$. The resulting connected components are called nodal domains. Dekel, Lee, and Linial proved that for fixed $p$, an eigenvector of a $G(n,p)$ random graph has almost all its vertices contained in two large nodal domains w.h.p. \cite{dekel2007eigenvectors}.  Arora and Bhaskara then showed that for $p\geq n^{-1/20}$ there are exactly 2 nodal domains  \cite{arora2011eigenvectors}. Huang and Rudelson proved that these domains are approximately the same size \cite{huang2020size}.

\paragraph{Green's function of the Cartesian product.} 
Chung and Yau, then Ellis \cite{chung2000discrete, ellis2003discrete} proved that the entries of the Green's function of a Cartesian product can be expressed as a contour integral of a function of Green's functions on the two original graphs, assuming that both graphs are finite.

\section{EXP is not sufficient for quantum ergodicity}\label{sec:onlyexp}

 For a graph $G=(V,E)$, we will use $|G|$ to denote the number of vertices $|V|$. $I_n$ refers to the identity operator of dimension $n$. Unless otherwise specified, for a graph $G$, we will use the notation $G=(V_G,E_G)$ with adjacency operator $\A_G$.  For $u,v\in V$, we will write $u\sim v$ to signify $(u,v)\in E$.

\begin{dfn}
Consider two graphs $G_1=(V_1,E_1)$ and $G_2=(V_2,E_2)$. The \emph{Cartesian product} $G_1\square G_2$ of the graphs $G_1$ and $G_2$ is defined as follows. $G_1\square G_2$ has vertex set $V_1\times V_2$, and for $(u_1,u_2),(v_1,v_2)\in V_1\times V_2$, $(u_1,u_2)\sim (v_1,v_2)$ if and only if either
\begin{enumerate}
    \item $u_1\sim v_1$ in $G_1$ and $u_2=v_2$ or
    \item $u_1=v_1$ and $u_2\sim v_2$ in $G_2$.
\end{enumerate}
\end{dfn}
An equivalent characterization is that if $G_1$ and $G_2$ have adjacency operations $\A_1$ and $\A_2$ respectively, then $G_1\square G_2$ is the graph with adjacency operator $\A_1\tensor I_{|G_2|}+I_{|G_1|}\tensor \A_2$.

Note that the Cartesian product is well defined for locally finite graphs (graphs where each vertex has finite degree), even if the graphs have an infinite number of vertices. However, in this section, we assume that all graphs are finite.

Given $\phi_i:V_i\rightarrow \R$ for $i\in\{1,2\}$, we define $\psi_{\phi_1,\phi_2}:V_1\times V_2\rightarrow \R$ to be 
\begin{equation}\label{eq:cartesianeig}
\psi_{\phi_1,\phi_2}(u_1,u_2):=\phi_1(u_1)\cdot\phi_2(u_2).\end{equation} The key property of Cartesian products we use is the following:

\begin{fact}\label{fact:cartesian}
If $\phi_1$ is an eigenvector of $\A_1$ of eigenvalue $\lambda_1$ and $\phi_2$ is an eigenvector of $\A_2$ with eigenvalue $\lambda_2$, then $\psi_{\phi_1,\phi_2}$ is an eigenvector of $\A_{G_1\square G_2}$ of eigenvalue $\lambda_1+\lambda_2$. 
\end{fact}

Take $(G_n)$ as a family of $d$-regular graphs $(V_n,E_n)$ that satisfies EXP with a fixed parameter $\epsilon>0$ with adjacency operators $(\A_n)$. Let $C_4$ denote the cycle graph of length 4. $G_n\square C_4$ is a $d+2$ regular graph on $4n$ vertices.

\begin{prop}
The family of graphs $(G_n\square C_4)$ satisfies EXP but does not satisfy BST. 
\end{prop}
\begin{proof}
As the spectrum of $\A_{C_4}$ is $\{2, 0, 0, -2\}$, by Fact \ref{fact:cartesian}, the adjacency operator of $G_n\square C_4$ satisfies EXP with parameter $\frac{\min\{d\epsilon,2\}}{d+2}$, which is constant for constant $d$.

Given $(u_1,u_2)\in V_n\times V_{C_4}$, take $(v_1,v_2)$ such that $u_1\sim v_1$ in $G_n$ and $u_2\sim v_2$ in $C_4$. This defines a 4$-$cycle in $G_n\square C_4$, given by
\[(u_1,u_2)\sim (u_1,v_2)\sim (v_1,v_2)\sim (v_1,u_2)\sim (u_1,u_2).\]
By the regularity of $G_n$ and $C_4$, such a vertex pair $(v_1,v_2)$ will exist for any $(u_1,u_2)$ (in fact, a total of $2d$ such pairs will exist). Therefore, there is a cycle of length $4$ starting at any given vertex, so if $R\geq 2$, then $\frac{|\{(u_1,u_2)\in V_n\times V_{C_4}, \rho(x)<R\}|}{4n}=1$. This means $G_n\square C_4$ does not satisfy BST. 
\end{proof}
\begin{thm}[Implies Theorem \ref{thm:mainexp}]\label{thm:rs}
Each graph in the family $(G_n\square C_4)$ admits an eigendecomposition that violates quantum ergodicity.
\end{thm}
\begin{proof}

We order and label the four vertices of $V_{C_4}$ $\{1,2,3,4\}$. The localized eigenbasis of $\A_{C_4}$ we will use is given by the following table:

\[
    \begin{tabular}{c|c}
  eigenvector & eigenvalue\\
    \hline
       $(\frac12,\frac12,\frac12,\frac12)$ & 2 \\
       $(\frac1{\sqrt2},0,-\frac1{\sqrt{2}},0)$ & 0 \\
       $(0,\frac1{\sqrt2},0,-\frac1{\sqrt{2}})$ & 0 \\
       $(\frac12,-\frac12,\frac12,-\frac12)$ & -2
    \end{tabular}.\]

We define $a_{n}:(V_n\times V_{C_4})\rightarrow \R$,
\[
a_{n}(u_1,u_2)=\left\{\begin{array}{cc}
     1&  u_2\in \{1,3\}\\
     -1& u_2\in \{2,4\}.
\end{array}\right.
\]

Because $\sum_{V_n\times V_{C_4}}a_n(u_1,u_2)=0$ and $\|a_n\|_\infty=1$, $a_n$ satisfies the conditions of Theorem \ref{thm:qe}. Define $\psi_{\phi,i}$ as the eigenvector from (\ref{eq:cartesianeig}) corresponding to the normalized eigenvector $\phi$ of $\A_n$ and the $i$th eigenvector of $C_4$ given in the table. By Fact \ref{fact:cartesian}, for any $u\in V_n$,

\[
\psi_{\phi,2}(u,2)=\psi_{\phi,2}(u,4)=\psi_{\phi,3}(u,1)=\psi_{\phi,3}(u,3)=0.
\]

Therefore 
\[
|\langle \psi_{\phi,2},a_{4n}\psi_{\phi,2}\rangle |=|\langle \psi_{\phi,3},a_{4n}\psi_{\phi,3}\rangle|=1
\]
and 
\[
|\langle \psi_{\phi,1},a_{4n}\psi_{\phi,1}\rangle|=|\langle \psi_{\phi,4},a_{4n}\psi_{\phi,4}\rangle|=0.
\]

For any $n$,
\[
\frac{1}{4n}\sum_{\phi, i} |\langle \psi_{\phi,i},a_{4n}\psi_{\phi,i} \rangle|^2=\frac12,
\]
meaning this family of eigenbases is not quantum ergodic.
\end{proof}

\section{Green's function on the infinite Cartesian product}
\label{sec:gfconvergence}

Our goal in this section is to show the Benjamini-Schramm limit of $(G_n\square C_4)$ from Section \ref{sec:onlyexp} has bounded imaginary part of the Green's function under the added assumption that $(G_n)$ satisfies BST. Therefore the requirement of BST cannot be generalized to this looser requirement on the Benjamini-Schramm limit. This relies on the Cartesian product commuting with both the Benjamini Schramm limit and, in a sense, with the Green's function itself.

Take $\mu$ to be a measure over isomorphism classes of rooted graphs, and $X$ to be a finite graph. The measure $\mu_{\square X}$ is defined over the same space as $\mu$, such that for any set $\Gamma$ of isomorphism classes,
\[
\mu_{\square X}(\Gamma)=\frac{1}{|X|}\sum_{v\in V_X}\mu(\{(G,o):\exists (H,o')\in \Gamma \textnormal{ s.t. } (G\square X,(o,v))\cong (H,o') \}).
\]

\begin{prop}\label{prop:bs}
If the Benjamini-Schramm limit of $(G_n)$ is $\mu$, then the Benjamini-Schramm limit of $(G_n\square X)$ is $\mu_{\square X}$.
\end{prop}

\begin{proof}
 Because $(G_n)$ converges to $\mu$, $\forall \epsilon>0$, $\exists N$ such that for $n>N$, the distribution of $1/\epsilon$ rooted balls in $G_n$ is within $\epsilon$ in any metrization of the weak topology of that of $\mu$. We denote by $B_r(G,o)$ the ball of radius $r$ around the root $o$ in $G$. We have that
 \[B_{1/\epsilon}(G_n\square X,(o,v))\cong B_{1/\epsilon}(B_{1/\epsilon}(G_n,o)\square X,(o,v)).\]
This is to say that the distribution over $1/\epsilon$ neighborhoods in $G_n\square X$ only depends on $G_n$ up to vertices of distance $1/\epsilon$. Therefore the distribution of $1/\epsilon$ balls with root $(u,v)$ in $G_n\square X$ obtained by sampling $u$ at random is within $\epsilon$ of the measure $\mu_{\square X}$ conditioned on the root being of the form $(\cdot,v)$ for specific $v\in V_X$. Sampling uniformly over $v\in V_X$ and sending $\epsilon\rightarrow 0$ gives the result. 
\end{proof}

Consider a graph $G$ with adjacency operator $\A$ and $z\in \C_+$. The Green's function $\G_G^z$ is the unique operator such that $(\A-z)\G_G^z=I_{|G|}$.

  \begin{thm}[Restatement of Theorem \ref{thm:green}]
  Take any (potentially infinite) graph $G_1$ and a finite graph $G_2$. Let $\psi_1,\ldots \psi_k$ be an orthonormal eigenbasis of $\A_2$ with eigenvalues $\lambda_1,\ldots, \lambda_k$.
  
We have

\[
\G_{G_1\square G_2}^{z}=\sum_{i=1}^{k} \G_{G_1}^{z-\lambda_i}\otimes \psi_i\psi_i^T.
\]
\end{thm}
Originally I wrote a proof for when $G_1=T_d$ which calculated entries of the Green's function recursively, similar to the proof of the Kesten-McKay measure using recursion (see for example Section 3 of \cite{warzel2013resonant}). The proof below was then sent to me by Mostafa Sabri, which generalizes to any $G_1$ and is less computationally intensive.

\begin{proof}
The adjacency operator of $G_1\square G_2$ is $\mathcal A_1\tensor I_{|G_2|}+I_{|G_1|}\tensor \mathcal A_2$, where $\A_1$ and $\A_2$ are the adjacency operators of $G_1$ and $G_2$ respectively.

\begin{eqnarray*}
(\A_1\tensor I_{|G_2|}+I_{|G_1|}\tensor\A_2-z)(\G_{G_1}^{z-\lambda_i}\otimes \psi_i\psi_i^T)&=&\A_1\G_{G_1}^{z-\lambda_i}\otimes \psi_i\psi_i^T+\G_{G_1}^{z-\lambda_i}\otimes \A_2\psi_i\psi_i^T-z(\G_{G_1}^{z-\lambda_i}\otimes \psi_i\psi_i^T)\\
&=&\A_1\G_{G_1}^{z-\lambda_i}\otimes \psi_i\psi_i^T+\lambda_i(\G_{G_1}^{z-\lambda_i}\otimes \psi_i\psi_i^T)-z(\G_{G_1}^{z-\lambda_i}\otimes \psi_i\psi_i^T)\\
&=&(\A_1+\lambda_i-z)\G_{G_1}^{z-\lambda_i}\otimes \psi_i\psi_i^T\\
&=&I_{|G_1|}\otimes \psi_i\psi_i^T.
\end{eqnarray*}
Therefore
\[
\sum_{i=1}^k (\A_1\tensor I_{|G_2|}+I_{|G_1|}\tensor\A_2-z)(\G_{G_1}^{z-\lambda_i}\otimes \psi_i\psi_i^T)=\sum_{i=1}^k I_{|G_1|}\tensor \psi_i\psi_i^T=I_{|G_1\square G_2|}
\]
as desired.

\end{proof}

\begin{thm}
 EXP and having bounded imaginary part in entries of the Green's function of the Benjamini-Schramm limit do not necessarily imply quantum ergodicity.
\end{thm}
  \begin{proof} Take the family $(G_n)$ to satisfy EXP and BST. $(G_n\square C_4)$ satisfies EXP and, by Proposition \ref{prop:bs}, has Benjamini-Schramm limit $T_d\square C_4$. The entries of $\G_{T_d}^z$ have bounded imaginary part for $z\in \C^+$, so by Theorem \ref{thm:green} so does $\G_{T_d\square C}^{z}$, as $\psi_i\psi_i^T$ has norm 1. However, by Theorem \ref{thm:rs}, $(G_n\square C_4)$ has a family of eigenbases that violates quantum ergodicity.
  \end{proof}

\section{BST is not sufficient for quantum ergodicity}\label{sec:onlybst}

Consider any family of $d$-regular graphs $(F_n)$ for $d$ even and $d\geq 8$ such that $(F_n)$ satisfies BST and $|F_n|=n$. We construct a family of graphs $(H_n)$ as follows. Delete an arbitrary edge of $F_n$, and call this new graph $F'_n$. Create $d/2$ copies of $F'_n$. Then add a vertex $v_n$, and add an edge from $v_n$ to each of the $d$ vertices of degree $d-1$, two for each copy of $F'_n$. Call this graph $H_n=(V_{H_n},E_{H_n})$.

\begin{figure}
    \centering
    \includegraphics[width=9cm]{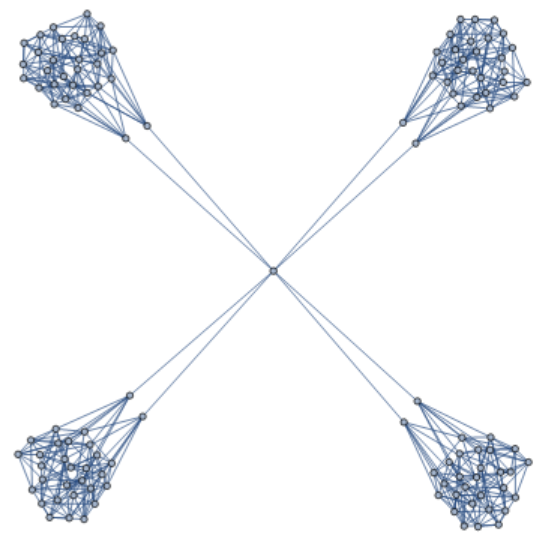}
    \caption{An example of a graph $H_n$ from Section \ref{sec:onlybst} drawn using Mathematica.}
    \label{fig:Hn}
\end{figure}

\begin{prop}
$(H_n)$ satisfies BST but not EXP. 
\end{prop}
\begin{proof}
 Take a vertex $x\in V_{H_n}\backslash v_n$. To differentiate between the injectivity radii of $H_n$ and $F_n$, we will write $\rho_{H_n}(x)$ and $\rho_{F_n}(x)$. The former refers to the injectivity radius of $x$ in $H_n$, whereas the latter is the injectivity radius of the vertex corresponding to $x$ in $F_n$. We claim that $\rho_{H_n}(x)\geq \rho_{F_n}(x)$. To see this, take a cycle through $x$ in $H_n$. If it intersects $v_n$, then it must intersect both neighbors of $v_n$ in the copy of $F_n$ that contains $x$. Therefore the length of such a cycle is at least $2\rho_{F_n}(x)+2$. If a cycle does not intersect $v_n$, it remains in $F_n$ and has length at least $2\rho_{F_n}(x)+1$. Putting these together, we have $\rho_{H_n}(x)\geq \rho_{F_n}(x)$. Therefore, $(H_n)$ satisfies BST.

By Cheeger's inequality (see for example \cite{chung1996laplacians}), because there are only $2$ edges from one copy of $F'_n$ to the rest of the graph, $\lambda_2\geq (1-4/n)d$, meaning $(H_n)$ does not satisfy EXP.
\end{proof}
\begin{thm}[Implies Theorem \ref{thm:mainbst}]
The family of graphs $(H_n)$ has an orthonormal eigenbasis which violates quantum ergodicity. 
\end{thm}
\begin{proof}
 Enumerate the copies of $F'_n$ in $H_n$ $F'_{n,1},\ldots, F'_{n,d/2}$. For any eigenvector $\phi$ of $\A_{F'_n}$, a normalized eigenvector $\chi$ of $\A_{H_n}$ of the same eigenvalue is given by
 \[
\chi(u)=\left\{ \begin{array}{cc}
      \phi(u)/\sqrt 2& u\in V_{F'_{n,1}} \\
      -\phi(u)/\sqrt 2& u\in V_{F'_{n,2}} \\
      0& \textnormal{otherwise}.
 \end{array}\right.
 \]
 Call $X$ the set of eigenvectors of this type. 
We then set

\[
a_n(u)=\left\{\begin{array}{cc}
     1&  u\in V_{F'_{n,1}},V_{F'_{n,2}}\\
     -1& u\in V_{F'_{n,3}},V_{F'_{n,4}}\\
     0& \textnormal{otherwise}.
     \end{array}\right.
\]
$a_n$ satisfies the conditions of a test function for quantum ergodicity. If we take $\Lambda$ to be an eigenbasis of $\A_{H_n}$ that contains $X$,

\[
\frac{1}{\frac d2n+1}\sum_{\psi\in \Lambda} |\langle \psi,a_n\psi \rangle|^2\geq \frac1{\frac d2n+1}\sum_{\chi\in X} |\langle\chi,a_n\chi \rangle|^2 =\frac1{\frac d2n+1}\sum_{\chi\in X} 1=\frac{n}{\frac{d}{2}n+1}\geq \frac1d
\]
violating quantum ergodicity.
\end{proof}

\subsection*{Acknowledgments}
I thank Nalini Anantharaman, Shirshendu Ganguly, Mostafa Sabri, and Nikhil Srivastava for numerous helpful discussions and comments on earlier versions of this manuscript.

\bibliographystyle{alpha}
\bibliography{ref}

\end{document}